\documentclass{article}
\usepackage{spconf,amsmath,graphicx}

\usepackage{cite}
\usepackage{amssymb,amsfonts}
\usepackage{algorithmic}
\usepackage{textcomp}
\usepackage{xcolor}
\def\BibTeX{{\rm B\kern-.05em{\sc i\kern-.025em b}\kern-.08em
    T\kern-.1667em\lower.7ex\hbox{E}\kern-.125emX}}

\usepackage[colorlinks=true,
            linkcolor=blue,
            citecolor=blue,
            urlcolor=blue]{hyperref}
\usepackage{amsthm}
\usepackage{bm}
\hypersetup{
    colorlinks=true,
    linkcolor=blue,
    filecolor=magenta,      
    urlcolor=cyan,
    pdftitle={MADUO},
    pdfpagemode=FullScreen,
    }
\usepackage{mathtools}
\usepackage{svg}
\usepackage{academicons}

\graphicspath{{fig/}}

\newtheorem{lemma}{Lemma}
\newtheorem{corollary}{Corollary}

\newcommand{\bmr}[1]{\bm{\mathrm{#1}}}
\newcommand{\NC}[1]{\mathcal{N}_{\mathbb{C}}\left(#1\right)}
\newcommand{\Expectation}[1]{\mathbb{E}\left\{#1\right\}}

\newcommand{\abs}[1]{\left|#1\right|}
\newcommand{\parenth}[1]{\left(#1\right)}
\newcommand{\MADUOscl}{MADUO$^{\text{scl}}$}
\newcommand{\z}{\widehat{\bmr{z}}}

\title{Master-Assisted Distributed Uplink Operation for Cell-Free Massive MIMO Networks}

\name{Andreas Angelou$^{\dagger}$ \qquad Pourya Behmandpoor$^{\star}$ \qquad Marc Moonen$^{\dagger}$ \thanks{This research was carried out at the ESAT Laboratory of KU Leuven, in the frame of Research Project FWO nr.~G0C0623N ``User-centric distributed signal processing algorithms for next generation cell-free massive MIMO based wireless communication networks''. The scientific responsibility is assumed by its authors. \\
}}
\address{$^{\dagger}$ Department of Electrical Engineering (ESAT), STADIUS, KU Leuven, Leuven, Belgium \\
$^{\star}$ Department of Electronics and Informatics, Vrije Universiteit Brussel (VUB), Brussels, Belgium \\
}

\begin{document}

\begin{figure*}[!t]
\centering
\section*{IEEE Copyright Notice}
{\copyright~2026 IEEE. Personal use of this material is permitted. Permission from IEEE must be obtained for all other uses, in any current or future media, including reprinting/republishing this material for advertising or promotional purposes, creating new collective works, for resale or redistribution to servers or lists, or reuse of any copyrighted component of this work in other works.}
\end{figure*}

\clearpage

\topmargin=0mm
\maketitle

\begin{abstract}
Cell-free massive multiple-input-multiple-output is considered a promising technology for the next generation of wireless communication networks. The main idea is to distribute a large number of access points (APs) in a geographical region to serve the user equipments (UEs) cooperatively. In the uplink, one of two types of operations is often adopted: centralized or distributed. In centralized operation, channel estimation and data decoding are performed at the central processing unit (CPU), whereas in distributed operation, channel estimation occurs at the APs and data detection at the CPU. In this paper, we propose a novel uplink operation, termed \emph{Master-Assisted Distributed Uplink Operation} (MADUO), where each UE is assigned a master AP, which receives soft data estimates from the other APs and decodes the data using its local signals and the received data estimates. Numerical experiments demonstrate that the proposed operation performs comparably to the centralized operation and balances fronthaul signaling and computational complexity.
\end{abstract}
\begin{keywords}
Cell-free massive MIMO, distributed processing, receive combining, uplink, spectral efficiency.
\end{keywords}
\section{Introduction}
\label{sec:intro}
Cellular massive multiple-input-multiple-output (MIMO) is a key enabler of 5G networks and has already reached commercial deployment \cite{marzetta2010noncooperative, bjornson2017massive, delson2019survey}. However, it suffers from inter-cell interference, especially for user equipments (UEs) near the cell edges, resulting in a non-uniform quality of service. Cell-free massive MIMO (CFmMIMO) \cite{ngo2015cell} addresses this bottleneck by deploying a large number of access points (APs) in the network to eliminate cell boundaries and cooperatively serve the UEs \cite{ngo2015cell}. CFmMIMO cancels interference \cite{interdonato2019ubiquitous} and provides high-quality service to all UEs \cite{ngo2017cell} in the network. As a result, it has drawn considerable attention and is considered a promising technology for 6G networks \cite{bjornson2019making, zhang2019cell, ngo2024ultradense, zheng2024mobile}.

In the uplink, one of two types of operations is often adopted: centralized or distributed. In centralized operation, APs forward pilot and data signals to the central processing unit (CPU), which performs channel estimation and data detection using channel state information (CSI) from all UEs (global CSI). In distributed operation, the APs use CSI of their served UEs (local CSI) and compute local data estimates (soft estimates) and send these to the CPU for final data detection using large-scale-fading decoding (LSFD). 

Various receive combining schemes exist for both operations. For centralized operation, the optimal Centralized minimum-mean-squared-error (C-MMSE) combiner \cite{bjornson2019making} uses global CSI, whereas the scalable alternative Partial MMSE (P-MMSE) \cite{bjornson2020scalable} offers scalability by suppressing interference from nearby UEs. For distributed operation, the APs use the Local MMSE (L-MMSE) \cite{bjornson2019making} combiner, which requires global CSI, whereas the Local Partial (LP-MMSE) \cite{demir2021foundations} requires only local CSI. In \cite{wang2025optimal}, a bilinear combiner reduces computational complexity by replacing the matrix inversion in MMSE combiners with a statistics-dependent matrix. In \cite{zheng2022team}, three Team MMSE (T-MMSE) combiners are proposed using the Theory of Teams and achieved near-optimal performance with reduced CSI sharing. A hybrid uplink scheme in \cite{kanno2022fronthaul} compresses signals at the APs before forwarding these to the CPU for centralized detection, reducing fronthaul signaling. In \cite{schulz2024scalable}, the LSFD coefficients are computed locally at the APs to further reduce fronthaul signaling.

Despite the efforts to reduce fronthaul signaling and computational load, uplink schemes that effectively balance these factors with performance remain limited. To this end, we propose a novel uplink scheme, termed \textit{Master-Assisted Distributed Uplink Operation} (MADUO), where each UE is assigned a Master AP (MAP) and a set of Additional Serving APs (ASAPs). The ASAPs estimate the signals of their served UEs and forward these, along with compressed CSI, to the MAPs for final data detection. Numerical experiments demonstrate that MADUO achieves near-centralized performance while balancing fronthaul signaling and computational load.

All simulation results can be reproduced using Matlab code available at: \href{https://github.com/AndreasAgg/Master-Assisted-Distributed-Uplink-Operation.git}{https://github.com/AndreasAgg/Master-Assisted-Distributed-Uplink-Operation.git}.

\section{Cell-Free Massive MIMO Model}
\label{sec:system-model}
We consider a CFmMIMO network with $K$ single-antenna UEs and $L$ APs, each equipped with $N$ antennas, with $LN \gg K$. We adopt the user-centric variant where each UE $k$ is served by a subset of APs, denoted by $\mathcal{A}_k$. The subset of UEs served by an AP $j$ is denoted by $\mathcal{U}_j$. The APs are connected to the CPU through ideal fronthaul links. We denote the channel between a UE $k$ and an AP $j$ by $\bmr{h}_{j,k} \in \mathbb{C}^N$. The collective channel of a UE $k$ to all APs is denoted by $\bmr{h}_k = \begin{bmatrix}
    \bmr{h}_{1,k}^T & \dots & \bmr{h}_{L,k}^T
\end{bmatrix}^T  \in \mathbb{C}^{LN \times 1}$, and the channels of all UEs to  an AP $j$ are collected in $\bmr{H}_j = \begin{bmatrix} \bmr{h}_{j,1} & \dots & \bmr{h}_{j,K} \end{bmatrix} \in \mathbb{C}^{N\times K}$. The channels are assumed to be constant over coherence blocks of $\tau_c$ samples. The pilot and uplink data transmission phases use $\tau_p$ and $\tau_u=\tau_c-\tau_p$ samples, respectively. 

\subsection{Channel Model and Channel Estimation}
We model the channel between an AP $j$ and a UE $k$ as a correlated Rayleigh fading channel $\bmr{h}_{j,k} \sim \mathcal{N}_{\mathbb{C}}\left(\bmr{0}_N, \bm{\mathrm{R}}_{j,k}\right)$, i.e., a zero-mean circularly symmetric complex Gaussian random vector with correlation matrix $\bm{\mathrm{R}}_{j,k} \in \mathbb{C}^{N\times N}$, which captures the large-scale fading. The correlation matrices are assumed constant over a large number of coherence blocks \cite{van2022gevd, van2023distributed}. 

We assume that each UE is assigned one $\tau_p$-samples pilot signal out of $\tau_p$ mutually orthogonal pilot signals; a pilot assignment algorithm is described in \cite[Sec.~V.A]{bjornson2020scalable}. We assume that the APs know the pilots of the UEs. In the pilot transmission phase, AP $j$ despreads its received signal with the pilot signal of UE $k$ to obtain the MMSE channel estimate \cite{bjornson2020scalable}
\begin{equation}
    \label{eq:channel-est}
    \widehat{\bmr{h}}_{j,k} = \sqrt{p_k \tau_p}\ \bm{\mathrm{R}}_{j,k} \bm{\Psi}_{j,t_k}^{-1} \bm{\mathrm{y}}_{j,t_k}^{\text{pilot}}\, ,
\end{equation}
where $
    \bm{\mathrm{y}}_{j,t_k}^{\text{pilot}} = \sum_{i\in\mathcal{P}_k} \sqrt{p_i \tau_p} \bmr{h}_{j,i} + \bm{\mathrm{n}}_{j,t_k}
$ is the received signal after despreading, with $\bm{\mathrm{n}}_{j,t_k} \sim \mathcal{N}_{\mathbb{C}}(\bm{0}_N, \sigma^2\bm{\mathrm{I}}_N)$, $t_k$ is the pilot index of UE $k$, $\mathcal{P}_k$ is the set of UEs with the same pilot signal as UE $k$ (including UE $k$ itself), $p_i$ is the transmission power of UE $i$, and $\bm{\Psi}_{j,t_k} \coloneqq \sum_{i\in\mathcal{P}_k} p_i\tau_p \bm{\mathrm{R}}_{j,i} + \sigma^2\bm{\mathrm{I}}_N$. We note that the channel estimation error $\widetilde{\bmr{h}}_{j,k} = \bmr{h}_{j,k} - \widehat{\bmr{h}}_{j,k}$ is independent of the channel estimate in \eqref{eq:channel-est} and is distributed as  $\widetilde{\bmr{h}}_{j,k} \sim \NC{\bm{0}_N, \bm{\mathrm{C}}_{j,k}}$, where $\bm{\mathrm{C}}_{j,k} =  \bm{\mathrm{R}}_{j,k} - p_k \tau_p \bm{\mathrm{R}}_{j,k} \bm{\Psi}_{j,t_k}^{-1} \bm{\mathrm{R}}_{j,k}$. Note that AP $j$ requires knowledge of the correlation matrices $\left\{\bm{\mathrm{R}}_{j,i} \right\}_{i \in \mathcal{P}_k}$ for channel estimation, which are assumed to be available; see \cite{bjornson2016massive, neumann2018covariance, upadhya2018covariance} for methods to estimate them. The channel estimates are collected into a matrix $\widehat{\bmr{H}}_j = \begin{bmatrix} \widehat{\bmr{h}}_{j,1} & \dots & \widehat{\bmr{h}}_{j,K} \end{bmatrix} \in \mathbb{C}^{N\times K}$.

\subsection{Data Transmission and Receive Combining} \label{subsec:Data-Transmission-and-Receive-Combining}
In the data transmission phase, the UEs send their data signals to the APs. The received signal $\bm{\mathrm{y}}_j \in \mathbb{C}^N$ for AP $j$ is 
\begin{equation} \label{eq:received-signal}
    \bm{\mathrm{y}}_j = \sum\limits_{i=1}^K \bmr{h}_{j,i}s_i + \bm{\mathrm{n}}_j \, ,
\end{equation}
where $s_i \sim \NC{0,p_i}$ denotes the data signal of UE $i$, and $\bm{\mathrm{n}}_j \sim \NC{\bm{0}_N,\sigma^2\bm{\mathrm{I}}_N}$ is the independent receiver noise. Here, we define the matrix $\bm{\mathrm{D}}_{j,k}$ which is the identity $\bm{\mathrm{I}}_N$ if UE $k$ is served by AP $j$ and all zeros matrix $\bm{0}_{N\times N}$ otherwise, and $\bm{\mathrm{D}}_k=\text{blkdiag}\left\{\bm{\mathrm{D}}_{j,k}\right\}_{j=1}^L$, where $\text{blkdiag}\{\bmr{X}_i\}_{i=1}^N$ denotes a block-diagonal matrix with blocks $\bmr{X}_1, \dots,\bmr{X}_N$.

In \textbf{centralized operation}, the APs forward their received signals to the CPU, which concatenates them to obtain a long vector $\bm{\mathrm{y}}  = \sum_{i=1}^K \bmr{h}_i s_i + \bm{\mathrm{n}} \in \mathbb{C}^{LN}$, where $\bm{\mathrm{y}}=\begin{bmatrix} \bm{\mathrm{y}}_1^T & \dots & \bm{\mathrm{y}}_L^T\end{bmatrix}^T$ and $\bm{\mathrm{n}}=\begin{bmatrix} \bm{\mathrm{n}}_1^T & \dots & \bm{\mathrm{n}}_L^T\end{bmatrix}^T$. The CPU then estimates the signal $s_k$ using the combiner $\bm{\mathrm{D}}_k\bm{\mathrm{v}}_k$ and obtains the signal estimate 
$
    \hat{s}_k = \bm{\mathrm{v}}_k^H \bm{\mathrm{D}}_k \bm{\mathrm{y}}
$.
The optimal linear combiner is the C-MMSE, which is defined in \cite[Eq.~(13)]{bjornson2019making} and uses global CSI, which makes it non-scalable as $K$ grows.\footnotemark \ A scalable alternative is to use CSI only from UEs that are served by the same APs as UE $k$, denoted by $\mathcal{S}_k$, to obtain the P-MMSE combiner
$\bm{\mathrm{v}}_k^\text{P-MMSE} = p_k \big( \sum_{i\in \mathcal{S}_k} p_i \bm{\mathrm{D}}_k (\widehat{\bmr{h}}_i \widehat{\bmr{h}}_i^H + \bm{\mathrm{C}}_i) \bm{\mathrm{D}}_k + \sigma^2 \bm{\mathrm{I}}_{LN}\big)^{-1} \bm{\mathrm{D}}_k \widehat{\bmr{h}}_{k}\, ,
$
where $\bmr{C}_i = \text{blkdiag}\{ \bmr{C}_{j,i}\}_{j=1}^L$ and $\widehat{\bmr{h}}_i = \begin{bmatrix}
    \widehat{\bmr{h}}_{1,i}^T & \dots &\widehat{\bmr{h}}_{L,i}^T
\end{bmatrix}^T$.
\footnotetext{We adopt the definition of scalability as presented in \cite[Definition 1]{bjornson2020scalable}.} 

In \textbf{distributed operation}, APs use a local receive combiner $\bm{\mathrm{D}}_{j,k} \bm{\mathrm{v}}_{j,k}$ to obtain the soft local estimates
\begin{equation} \label{eq:local-signal-estimates}
    \hat{s}_{j,k} = \bm{\mathrm{v}}_{j,k}^H \bm{\mathrm{D}}_{j,k} \bm{\mathrm{y}}_j\, .
\end{equation}
The use of $\bm{\mathrm{D}}_{j,k}$ ensures that AP $j$ only estimates the signals of the UEs $k\in\mathcal{U}_j$. APs obtain the signal estimate \eqref{eq:local-signal-estimates} using L-MMSE given by \cite[Eq.~(16)]{bjornson2019making}, which uses the estimated global CSI, i.e., $\{\widehat{\bmr{h}}_{j,i}\}_{i=1}^K$. A scalable alternative is the LP-MMSE
$
    \bm{\mathrm{v}}_{j,k}^{\text{LP-MMSE}} = p_k ( \sum_{i\in \mathcal{U}_j} p_i (\widehat{\bmr{h}}_{j,i} \widehat{\bmr{h}}_{j,i}^H + $ $\bm{\mathrm{C}}_{j,i}) + \sigma^2 \bm{\mathrm{I}}_N)^{-1} \bm{\mathrm{D}}_{j,k} \widehat{\bmr{h}}_{j,k}\, ,
$
which uses only local estimated CSI, i.e., $\{\widehat{\bmr{h}}_{j,i}\}_{i\in \mathcal{U}_j}$. The estimates are then passed on to the CPU, which combines them to obtain the final data estimate 
$
    \hat{s}_k = \sum_{j=1}^L a_{j,k}^* \hat{s}_{j,k}
$.
For this, the CPU computes the optimal weights $\{a_{j,k}\}_{j=1}^L$ in a non-scalable manner using \cite[Eq.~(5.30)]{demir2021foundations} or the nearly optimal weights in a scalable manner using \cite[Eq.~(5.41)]{demir2021foundations}.

\section{Master-Assisted Distributed Uplink Operation} \label{sec:MADUO}
We propose a novel uplink operation for CFmMIMO networks, called \textit{Master-Assisted Distributed Uplink Operation} (MADUO). In MADUO, every UE is assigned a Master AP (MAP). We denote the MAP of a UE $k$ by $l_k$ or $l$ when the UE in question is easily inferred. Note that $l_k \in \mathcal{A}_k$. We call \textit{Additional Serving APs} (ASAPs) of UE $k$ the APs in the set $\mathcal{A}_k \backslash  \{l_k\}$. For simplicity, we assume that $p_k=p$ for all $k$, although the analysis can be easily extended for UEs with different powers. In MADUO, the APs in $\mathcal{A}_k$ estimate the channels of all UEs in the network. ASAPs estimate the data signal of UE $k$ using the L-MMSE combiner $\bmr{D}_{j,k}\bmr{v}_{j,k}$ and then they send their soft data estimates $\hat{s}_{j,k}$, their fused channel estimates $\{ \bmr{v}_{j,k}^H \bmr{D}_{j,k} \widehat{\bmr{h}}_{j,i} \}_{i=1}^K$ and the Hermitian product 
\begin{equation} \label{eq:Herm}
    \mu_{j,k} = \bmr{v}_{j,k}^H \bmr{D}_{j,k}  \parenth{p\sum_{i=1}^K \bmr{C}_{j,i} + \sigma^2\bmr{I}_N}  \bmr{D}_{j,k} \bmr{v}_{j,k} 
\end{equation}
to MAP $l_k$. The Hermitian product in \eqref{eq:Herm} depends on the sum of correlation matrices associated with the channel estimation errors, and therefore reflects the quality of the channel estimates at AP $j$. Lastly, MAP $l_k$ estimates the data signal of UE $k$ using its local signals $\bmr{y}_{l_k}$ and the signal estimates $\hat{s}_{j,k}$ from the ASAPs as described in \hyperref[subsec:RC]{Section \ref{subsec:RC}}. We note that whenever we estimate the data signal of UE $k$, the APs in $\mathcal{A}_k$ are connected through fronthaul with a star topology, with the MAP acting as the central node connected to all ASAPs.

\subsection{Master AP Receive Combining}
\label{subsec:RC}
The MAP $l$ of UE $k$ has access to its own local signals $\bmr{y}_l$ and the signal estimates from the ASAPs $\{\hat{s}_{j,k}\}_{j \neq l}$. Hence, it has access to the $(N+L-1)$-dimensional vector
\begin{equation}
    \widetilde{\bmr{y}}_l = \begin{bmatrix}
        \bmr{y}_l^T &
        \hat{s}_{1,k} &
        \dots &
        \hat{s}_{l-1,k} &
        \hat{s}_{l+1,k} &
        \dots &
        \hat{s}_{L,k}
    \end{bmatrix}^T\, ,
\end{equation}
with $\bmr{y}_l$ given by \eqref{eq:received-signal} and the signal estimates given by \eqref{eq:local-signal-estimates}. The receive combining vector of the MAP is 
\begin{equation} \label{eq:MADUO-RC-init}
    \bmr{v}_k = \begin{bmatrix}
        \bmr{v}_{l,k}^T & \bmr{a}_k^T
    \end{bmatrix}^T\, ,
\end{equation}
where $\bmr{v}_{l,k} \in \mathbb{C}^N$ will combine the local signals $\bmr{y}_l$, and $\bmr{a}_k = \begin{bmatrix}
    a_{1,k} & \dots & a_{l-1,k} & a_{l+1,k} & \dots & a_{L,k}
\end{bmatrix}^T \in \mathbb{C}^{L-1}$ will act on the signal estimates. The final signal estimate for UE $k$ is 
\begin{equation} \label{eq:MADUO-signal-estimate}
    \hat{s}_k = \bmr{v}_k^H \widetilde{\bmr{y}}_l = \bmr{v}_{l,k}^H \bmr{y}_l + \sum\limits_{j=1,\ j\neq l}^L a_{j,k}^* \hat{s}_{j,k}\, .
\end{equation}
We aim to compute the receive combiner $\bmr{v}_k$ in \eqref{eq:MADUO-RC-init} by maximizing the achievable spectral efficiency (SE).

\begin{lemma} \label{lemma}
    An achievable SE of UE $k$ in MADUO is
    \begin{equation} \label{eq:MADUO-SE}
        \text{SE}_k = \frac{\tau_u}{\tau_c} \Expectation{\log_2\left( 1 + \text{SINR}_k\right)}\, ,
    \end{equation}
    where the expectation is with respect to channel realizations and the effective SINR is 
    \begin{equation} \label{eq:MADUO-SINR}
        \text{SINR}_k = \frac{p \abs{\bmr{v}_k^H \z_{l,kk}}^2}{\bmr{v}_k^H \bmr{B}_k\bmr{v}_k}\, ,
    \end{equation}
    with $\z_{l,kk} =\begin{bmatrix} \widehat{\bmr{h}}_{l,k} \\ \widehat{\bmr{g}}_{kk} \end{bmatrix}$, where
    $
    \widehat{\bmr{g}}_{kk}[r,:] = \bmr{v}_{r,k}^H \bmr{D}_{r,k} \widehat{\bmr{h}}_{r,k}
    $
    is the $r$-th row of $\widehat{\bmr{g}}_{kk}\in \mathbb{C}^{L-1}$ (the row $r=l$ is excluded). Moreover, 
\begin{equation}
\begin{split} \label{eq:B}
    \bmr{B}_k & = \\ & \begin{bmatrix}
        p \parenth{
        \widehat{\bmr{H}}_{l,-k}\widehat{\bmr{H}}_{l,-k}^H + \sum\limits_{i=1}^K \bmr{C}_{l,i}
        } + \sigma^2 \bmr{I}_N & 
        p\widehat{\bmr{H}}_{l,-k} \widehat{\bmr{G}}_k^H \\
        p\widehat{\bmr{G}}_k \widehat{\bmr{H}}_{l,-k}^H &
        p \widehat{\bmr{G}}_k\widehat{\bmr{G}}_k^H + \bmr{F}_k
    \end{bmatrix}
\end{split}\, 
\end{equation}
    where $\widehat{\bmr{H}}_{l,-k}\in \mathbb{C}^{N\times(K-1)}$ is the same as $\widehat{\bmr{H}}_{l}$ but without the column $\widehat{\bmr{h}}_{l,k}$.
    $\widehat{\bmr{G}}_k \in \mathbb{C}^{(L-1) \times (K-1)}$ with
    $\widehat{\bmr{G}}_k[r,:] = \bmr{v}_{r,k}^H \bmr{D}_{r,k} \widehat{\bmr{H}}_{r,-k}$
(the row $r=l$ is excluded) and $\bmr{F}_k = \text{diag}\{\mu_{j,k}\}_{j\neq l} \in \mathbb{C}^{(L-1)\times (L-1)}$, where $\text{diag}\{x_i\}_{i=1}^N$ denotes a diagonal matrix with diagonal element $x_1,\dots,x_N$.
\end{lemma}
\begin{proof}
We first expand \eqref{eq:MADUO-signal-estimate} using \eqref{eq:received-signal}, \eqref{eq:local-signal-estimates}, and $\bmr{h}_{j,k} = \widetilde{\bmr{h}}_{j,k} + \widehat{\bmr{h}}_{j,k}$, and then we use \cite[Lem.~3.5]{demir2021foundations} by identifying the effective channel and the interference. We make use of the independence of the channel estimates and estimation errors, and the independence of the data signals to derive \eqref{eq:MADUO-SINR}.
\end{proof}
\noindent
Since \eqref{eq:MADUO-SINR} has the form of a generalized Rayleigh quotient, the receive combiner that maximizes \eqref{eq:MADUO-SINR} is given as follows:
\begin{corollary}
    The effective SINR in \eqref{eq:MADUO-SINR} is maximized by 
    \begin{equation} \label{eq:MADUO-RC}
             \bmr{v}_k = \bmr{B}_k ^{-1} \widehat{\bmr{z}}_{l,kk}\, .
        \end{equation}
\end{corollary}
\begin{proof}
    The proof follows from \cite[Lem.~B.10]{bjornson2017massive} since we maximize a generalized Rayleigh quotient with respect to $\bmr{v}_k$. 
\end{proof}

\subsection{Scalable MADUO}
To make the proposed operation scalable, we make the following modifications. APs in $\mathcal{A}_k$ estimate only the channels of their served UEs. ASAPs use LP-MMSE instead of L-MMSE for their soft data estimates $\hat{s}_{j,k}$, and send only the fused channel estimates of their served UEs $\{ \bmr{v}_{j,k}^H \bmr{D}_{j,k} \widehat{\bmr{h}}_{j,i} \}_{i\in \mathcal{U}_j}$ along with
\begin{equation}
    \mu_{j,k}^{\prime} = \bmr{v}_{j,k}^H \bmr{D}_{j,k}  \parenth{p\sum_{i\in\mathcal{U}_j} \bmr{C}_{j,i} + \sigma^2\bmr{I}_N}  \bmr{D}_{j,k} \bmr{v}_{j,k} 
\end{equation}
where the sum includes only the UEs $i \in \mathcal{U}_j$. Lastly, the MAP uses only the channel estimates of its served UEs and the information from the ASAPs. Thus, the sum in the top-left block of \eqref{eq:B} involves only the UEs $i\in \mathcal{U}_l$ and $\widehat{\bmr{H}}_{l,-k}$ contains non-zero columns only for UEs $i\in \mathcal{U}_l \backslash \{k\}$. Hereafter, the scalable MADUO will be denoted by \MADUOscl.

\subsection{Fronthaul Signaling and Computational Complexity}
The number of complex scalars exchanged through the fronthaul for MADUO and \MADUOscl \ per coherence block is
\begin{equation*}
    \sum_{j=1}^L \left(\tau_u + \Omega + 1\right) \left(\left|\mathcal{U}_j\right| - \left| \mathcal{U}^{\text{master}}_j \right| \right)\, ,
\end{equation*}
where $\Omega=K$ for MADUO and $\Omega=\abs{\mathcal{U}_j}$ for \MADUOscl, and $\mathcal{U}^{\text{master}}_j$ is the set of UEs for which AP $j$ acts as MAP.

The number of complex multiplications for the computation of the receive combiner in \eqref{eq:MADUO-RC} is $\mathcal{O}\parenth{(N+\abs{\mathcal{A}_k}-1)^3}$ for MADUO and \MADUOscl. However, the number of computations in MADUO is directly dependent on $K$, since the APs estimate and compress the channels to all $K$ UEs.

\section{Numerical Experiments}
We compare MADUO and \MADUOscl \ with the centralized and distributed operation in terms of SE performance, fronthaul signaling, and computational complexity, using numerical simulations. For the centralized operation, we consider C-MMSE \cite{bjornson2019making} and P-MMSE combining \cite{bjornson2020scalable}. For the distributed operation, we consider L-MMSE \cite{bjornson2019making} and LP-MMSE \cite{demir2021foundations}. We consider a CFmMIMO network with $L=100$ APs, $N=4$ antennas per AP, and varying $K$, in a simulation square area of $2\times2 \ \text{km}^2$. The network has a wrap-around topology, as explained in \cite[Sec.~5.3]{demir2021foundations}. The cluster formation, master AP assignment, and pilot assignment are jointly implemented, as described in \cite[Sec.~V.A]{bjornson2020scalable}. Moreover, we use the same propagation model as in \cite[Sec.~5.3]{bjornson2017massive}, and $\tau_c=200$, $\tau_p=10$, $p_k=p=100$ mW. 

\begin{figure}[!t]
\centering
\def\svgwidth{0.93\columnwidth}
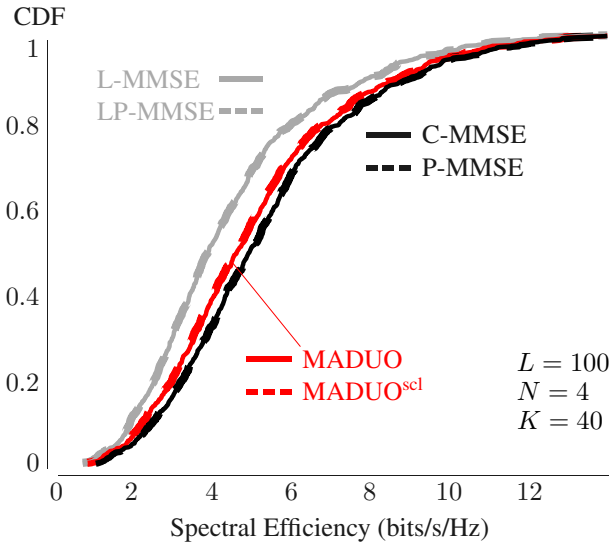
\caption{Cumulative distribution function (CDF) of the per-UE SE for $L=100$, $N=4$, and $K=40$. MADUO and \MADUOscl \ perform comparably to the centralized operation.}
\label{fig:scenario}
\end{figure}
In \hyperref[fig:scenario]{Fig.~\ref{fig:scenario}}, we plot the cumulative distribution function (CDF) of the per-UE SE. \MADUOscl\ matches the performance of MADUO, and both significantly outperform the distributed operation and closely approach the centralized operation. In \hyperref[fig:fronthaul]{Fig.~\ref{fig:fronthaul}}, we plot the number of complex scalars sent through fronthaul per coherence block for varying $K$. For $K>60$, \MADUOscl \ reduces the fronthaul signaling compared to the distributed operation as $K$ increases, whereas MADUO requires the most fronthaul signaling. In \hyperref[fig:computation]{Fig.~\ref{fig:computation}}, we plot the number of complex multiplications required to estimate the signal of one UE per coherence block, for varying $K$. MADUO and \MADUOscl \ require fewer computations than the centralized operation. Interestingly, the number of computations of \MADUOscl \ decreases as $K$ grows. The reason is that fewer APs will serve an arbitrary UE $k$ as $K$ increases, i.e., $\abs{\mathcal{A}_k}$ will decrease with $K$ \cite{demir2021foundations}.

\begin{figure}[!t]
\centering
\def\svgwidth{0.93\columnwidth}
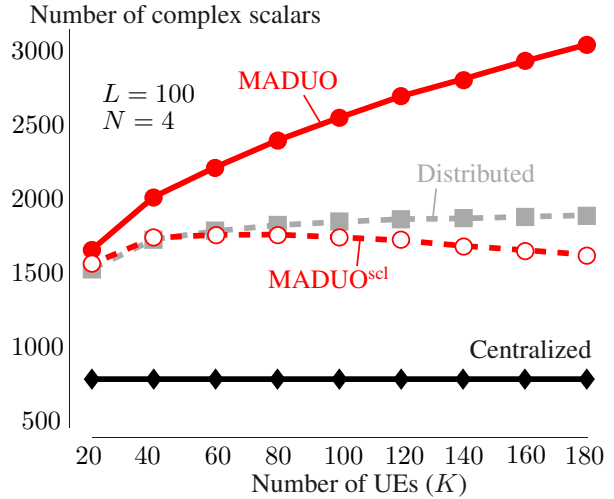
\caption{Fronthaul signaling for $L=100$, $N=4$, and varying $K$. For $K>60$, \MADUOscl \ requires less fronthaul signaling than the distributed operation.}
\label{fig:fronthaul}
\end{figure}

\begin{figure}[!t]
\centering
\def\svgwidth{0.93\columnwidth}
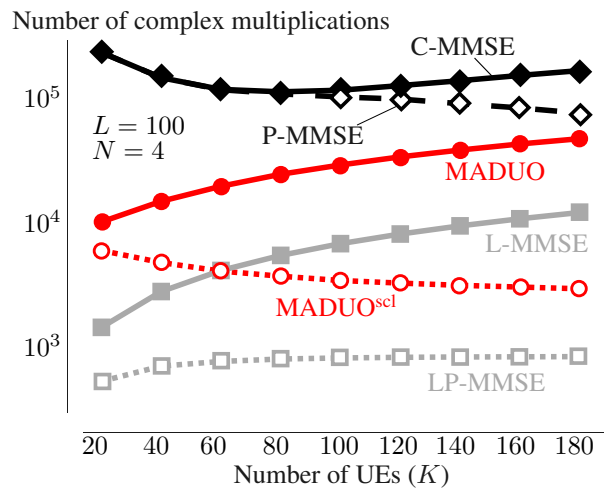
\caption{Number of complex multiplications for the service of one UE for varying $K$, with $L=100$ and $N=4$. The number of computations of \MADUOscl \ decreases with $K$, approaching that of the distributed operation.}
\label{fig:computation}
\end{figure}

\section{Conclusion}
\label{sec:conclusion}
In this paper, we have proposed MADUO, a novel distributed uplink operation for CFmMIMO. In MADUO, each UE is assigned a MAP and a set of ASAPs, which are responsible for its service. The MAP uses its local signals and the signal estimates from the ASAPs for data detection.  MADUO is non-scalable, and we have proposed modifications to make it scalable and obtain \MADUOscl. Both MADUO and \MADUOscl \ perform comparably to the centralized operation, while requiring fewer computations. Moreover, \MADUOscl \ reduces the fronthaul signaling compared to the distributed operation.

\bibliographystyle{IEEEbib}
\bibliography{strings,refs}

@INPROCEEDINGS{delson2019survey,
  author={T.R., Delson and Jose, Iven},
  booktitle={2019 International Conference on Data Science and Communication (IconDSC)}, 
  title={A Survey on 5G Standards, Specifications and Massive MIMO Testbed Including Transceiver Design Models Using QAM Modulation Schemes}, 
  year={2019},
  volume={},
  number={},
  pages={1-7},
  doi={10.1109/IconDSC.2019.8816942}}

@article{ngo2024ultradense,
  title={Ultradense cell-free massive MIMO for 6G: Technical overview and open questions},
  author={Ngo, Hien Quoc and Interdonato, Giovanni and Larsson, Erik G and Caire, Giuseppe and Andrews, Jeffrey G},
  journal={Proceedings of the IEEE},
  year={2024},
  publisher={IEEE}
}

@article{marzetta2010noncooperative,
  title={Noncooperative cellular wireless with unlimited numbers of base station antennas},
  author={Marzetta, Thomas L},
  journal={IEEE transactions on wireless communications},
  volume={9},
  number={11},
  pages={3590--3600},
  year={2010},
  publisher={IEEE}
}

@article{bjornson2019making,
  title={Making cell-free massive MIMO competitive with MMSE processing and centralized implementation},
  author={Bj{\"o}rnson, Emil and Sanguinetti, Luca},
  journal={IEEE Transactions on Wireless Communications},
  volume={19},
  number={1},
  pages={77--90},
  year={2019},
  publisher={IEEE}
}

@article{bjornson2017massive,
  title={Massive MIMO networks: Spectral, energy, and hardware efficiency},
  author={Bj{\"o}rnson, Emil and Hoydis, Jakob and Sanguinetti, Luca and others},
  journal={Foundations and Trends{\textregistered} in Signal Processing},
  volume={11},
  number={3-4},
  pages={154--655},
  year={2017},
  publisher={Now Publishers, Inc.}
}

@inproceedings{ngo2015cell,
  title={Cell-free massive MIMO: Uniformly great service for everyone},
  author={Ngo, Hien Quoc and Ashikhmin, Alexei and Yang, Hong and Larsson, Erik G and Marzetta, Thomas L},
  booktitle={2015 IEEE 16th international workshop on signal processing advances in wireless communications (SPAWC)},
  pages={201--205},
  year={2015},
  organization={IEEE}
}

@ARTICLE{ngo2017cell,
  title={Cell-Free Massive MIMO Versus Small Cells},
  author={Ngo, Hien Quoc and Ashikhmin, Alexei and Yang, Hong and Larsson, Erik G. and Marzetta, Thomas L.},
  journal={IEEE Transactions on Wireless Communications}, 
  year={2017},
  volume={16},
  number={3},
  pages={1834-1850},
  doi={10.1109/TWC.2017.2655515}
}

@article{demir2021foundations,
  title={Foundations of user-centric cell-free massive MIMO},
  author={Demir, {\"O}zlem Tugfe and Bj{\"o}rnson, Emil and Sanguinetti, Luca and others},
  journal={Foundations and Trends{\textregistered} in Signal Processing},
  volume={14},
  number={3-4},
  pages={162--472},
  year={2021},
  publisher={Now Publishers, Inc.}
}

@article{bjornson2020scalable,
  title={Scalable cell-free massive MIMO systems},
  author={Bj{\"o}rnson, Emil and Sanguinetti, Luca},
  journal={IEEE Transactions on Communications},
  volume={68},
  number={7},
  pages={4247--4261},
  year={2020},
  publisher={IEEE}
}

@article{interdonato2019ubiquitous,
  title={Ubiquitous cell-free massive MIMO communications},
  author={Interdonato, Giovanni and Bj{\"o}rnson, Emil and Quoc Ngo, Hien and Frenger, P{\aa}l and Larsson, Erik G},
  journal={EURASIP Journal on Wireless Communications and Networking},
  volume={2019},
  number={1},
  pages={1--13},
  year={2019},
  publisher={Springer}
}

@article{wang2025optimal,
  title={Optimal bilinear equalizer for cell-free massive MIMO systems over correlated Rician channels},
  author={Wang, Zhe and Zhang, Jiayi and Bj{\"o}rnson, Emil and Niyato, Dusit and Ai, Bo},
  journal={IEEE Transactions on Signal Processing},
  year={2025},
  publisher={IEEE}
}

@inproceedings{zheng2022team,
  title={Team-optimal MMSE combining for cell-free massive MIMO systems},
  author={Zheng, Jiakang and Zhang, Jiayi and Ai, Bo},
  booktitle={ICC 2022-IEEE International Conference on Communications},
  pages={1306--1311},
  year={2022},
  organization={IEEE}
}

@inproceedings{kanno2022fronthaul,
  title={Fronthaul Load-Reduced Scalable Cell-Free massive MIMO by Uplink Hybrid Signal Processing},
  author={Kanno, Issei and Ito, Masaaki and Ohseki, Takeo and Yamazaki, Kosuke and Kishi, Yoji and Choi, Thomas and Chen, Wei-Yu and Molisch, Andreas F},
  booktitle={2022 IEEE 95th Vehicular Technology Conference:(VTC2022-Spring)},
  pages={1--5},
  year={2022},
  organization={IEEE}
}

@article{zhang2019cell,
  title={Cell-free massive MIMO: A new next-generation paradigm},
  author={Zhang, Jiayi and Chen, Shuaifei and Lin, Yan and Zheng, Jiakang and Ai, Bo and Hanzo, Lajos},
  journal={IEEE access},
  volume={7},
  pages={99878--99888},
  year={2019},
  publisher={IEEE}
}

@article{zheng2024mobile,
  title={Mobile cell-free massive MIMO: Challenges, solutions, and future directions},
  author={Zheng, Jiakang and Zhang, Jiayi and Du, Hongyang and Niyato, Dusit and Ai, Bo and Debbah, M{\'e}rouane and Letaief, Khaled B},
  journal={IEEE Wireless Communications},
  volume={31},
  number={3},
  pages={140--147},
  year={2024},
  publisher={IEEE}
}

@inproceedings{van2022gevd,
  title={GEVD-based low-rank channel covariance matrix estimation and MMSE channel estimation for uplink cellular massive MIMO systems},
  author={Van Rompaey, Robbe and Moonen, Marc},
  booktitle={2022 30th European Signal Processing Conference (EUSIPCO)},
  pages={1641--1645},
  year={2022},
  organization={IEEE}
}

@inproceedings{schulz2024scalable,
  title={Scalable Cell-Free Massive MIMO with Fully Distributed Large-Scale Fading Decoding},
  author={Schulz, Leonard Paul and Schappmann, Christian and Bauch, Gerhard},
  booktitle={2024 19th International Symposium on Wireless Communication Systems (ISWCS)},
  pages={1--6},
  year={2024},
  organization={IEEE}
}

@article{van2023distributed,
  title={Distributed MMSE-based uplink receive combining, downlink transmit precoding and optimal power allocation in cell-free massive MIMO systems},
  author={Van Rompaey, Robbe and Moonen, Marc},
  journal={EURASIP Journal on Wireless Communications and Networking},
  volume={2023},
  number={1},
  pages={48},
  year={2023},
  publisher={Springer}
}

@inproceedings{bjornson2016massive,
  title={Massive MIMO with imperfect channel covariance information},
  author={Bj{\"o}rnson, Emil and Sanguinetti, Luca and Debbah, Merouane},
  booktitle={2016 50th Asilomar Conference on Signals, Systems and Computers},
  pages={974--978},
  year={2016},
  organization={IEEE}
}

@article{neumann2018covariance,
  title={Covariance matrix estimation in massive MIMO},
  author={Neumann, David and Joham, Michael and Utschick, Wolfgang},
  journal={IEEE Signal Processing Letters},
  volume={25},
  number={6},
  pages={863--867},
  year={2018},
  publisher={IEEE}
}

@article{upadhya2018covariance,
  title={Covariance matrix estimation for massive MIMO},
  author={Upadhya, Karthik and Vorobyov, Sergiy A},
  journal={IEEE Signal Processing Letters},
  volume={25},
  number={4},
  pages={546--550},
  year={2018},
  publisher={IEEE}
}

\end{document}